\documentclass{article}

\usepackage[english]{babel}
\usepackage{enumerate}
\usepackage{amsthm} 

\usepackage[a4paper,top=3cm,bottom=3cm,left=3cm,right=3cm,marginparwidth=1.75cm]{geometry}

\usepackage[round]{natbib}
\usepackage{pifont}

\usepackage{amsmath}

\usepackage{amsfonts}
\usepackage{bbm}
\usepackage{graphicx}
\usepackage[colorlinks=true, allcolors=blue]{hyperref}
\usepackage{mathtools}
\usepackage{breqn}
\usepackage{float}
\usepackage{adjustbox}	
\usepackage{xcolor}
\usepackage{makecell}
\usepackage{subfigure}
\usepackage{caption}

\newtheorem{theorem}{Theorem}[section]

\newtheorem{proposition}[theorem]{Proposition}

\newtheorem{remark}[theorem]{Remark}
\newtheorem{example}[theorem]{Example}

\numberwithin{equation}{section}

\def\E{{\mathbb{E}}}
\newcommand{\R}{{\mathbb R}}

\definecolor{blue0}{RGB}{0,77,153} 
\definecolor{red0}{RGB}{179,0,77} 
\definecolor{green0}{RGB}{134,219,76} 
\definecolor{gray0}{RGB}{84,97,110}

\usepackage{mathtools}
\usepackage{authblk}
\mathtoolsset{showonlyrefs}
\usepackage[normalem]{ulem}
\usepackage{todonotes}

\title{Capturing Smile Dynamics with the Quintic Volatility Model: SPX, Skew-Stickiness Ratio and VIX}

\author[1]{Eduardo Abi Jaber\thanks{eduardo.abi-jaber@polytechnique.edu. The first author is grateful for the financial support from the Chaires FiME-FDD, Financial Risks, Deep Finance \& Statistics and Machine Learning and systematic methods in finance at École Polytechnique.}}
\author[2]{Shaun (Xiaoyuan) Li\thanks{shaunlinz02@gmail.com. The views expressed in this paper are those of the authors and do not reflect those of Morgan Stanley or its affiliates.\\} }
\affil[1]{Ecole Polytechnique, CMAP}
\affil[2]{Morgan Stanley, Paris}

\begin{document}

\maketitle

\begin{abstract}
We introduce the two-factor Quintic Ornstein–Uhlenbeck (OU) model, where volatility is modelled as a degree-five polynomial of the sum of two Ornstein–Uhlenbeck processes driven by the same Brownian motion, each mean-reverting at a different speed. We demonstrate that the model effectively captures the volatility surfaces of SPX and VIX while aligning with the skew-stickiness ratio (SSR) across maturities ranging from a few days to over two years. Furthermore, it is consistent with key empirical stylized facts, notably reproducing the Zumbach effect.
\end{abstract}

\begin{description}
\item[JEL Classification:] G13, C63, G10, C45. 
\item[Keywords:] SPX \& VIX options, Stochastic volatility, Skew-Stickiness-Ratio,  Derivative pricing, Calibration, SPX Stylized facts, Quantization
\end{description}

\section{Introduction}
A good volatility model must capture not only the shape of the volatility smile, but also its dynamics. This is of particular importance for SPX options: while many models can reproduce the shape of the volatility surface at a given point in time, not many of them can correctly describe how the volatility surface evolves.

A key measure of this evolution is the skew-stickiness ratio (SSR), introduced by \cite*{bergomi2015stochastic}, which quantifies the sensitivity of implied volatility to log-price variations. This metric is widely used in the industry for options trading, hedging, and price discovery, making it a crucial test for any volatility model. However, despite its importance, fitting the SPX volatility surface while remaining consistent with the SSR remains a challenging task for existing models in the literature.

This limitation has been highlighted in recent works such as \cite*{bourgey2024smile, friz2024computing}, which emphasize the inherent difficulties in capturing implied volatility dynamics. In particular,  \cite*{friz2024computing} shows that models with similar quality of fit to the SPX smile can produce drastically different SSR term structure, while \cite*{bourgey2024smile} explores the trade-offs between calibration quality and dynamic consistency. Specifically, \cite*{bourgey2024smile} investigates the calibration of several stochastic volatility models—including Heston, two-factor Bergomi, and their rough counterparts—against the term structure of SPX ATM implied volatility, skew, and SSR. The study finds that even after careful calibration, the SSR of these models “significantly deviates from empirical market behavior around the calibration date”, revealing a structural inconsistency between these stochastic volatility models and the observed dynamics of the SPX volatility surface.

To quote  \cite*{friz2024computing} directly:
\begin{center}
\textit{``This inconsistency remains, in our view, an important outstanding problem.}''    
\end{center}

We introduce the two-factor Quintic Ornstein–Uhlenbeck (OU) model, which resolves this issue. In this model, the volatility is modelled as a polynomial of degree five of the sum of two OU processes driven by the same Brownian Motion, each mean-reverting at a different speed. This structure provides the flexibility needed to simultaneously capture the term structure of SPX volatilities, at-the-money skew, and SSR, see Figure~\ref{fig:spx_term_struc}, overcoming the limitations of existing stochastic volatility models.

Our model extends the one-factor Quintic OU model introduced by \cite*{jaber2022quintic}, originally designed to fit the joint SPX-VIX volatility surfaces for maturities up to three months. The one-factor Quintic OU model achieved excellent fits across 10 years of data (see \cite*{abi2022joint}) but lacked the flexibility to fit longer maturities without introducing time-dependent parameters. By incorporating a second factor, the two-factor Quintic OU model provides better control over the volatility term structure, allowing for a more robust joint calibration of SPX and VIX smiles across different maturities ranging from a few days to over two years, as shown in Figure~\ref{fig:calib_2_factor}.

In addition, we take a step forward by calibrating both the SPX-VIX volatility surfaces and the term structure of the SSR simultaneously - something that has not been attempted before in the literature. Joint calibration of SPX-VIX surfaces is already notoriously challenging on its own, let alone adding the SPX SSR term structures to the mix. Nonetheless, our model succeeds in capturing SPX-VIX volatility surfaces, while remaining consistent with the observed SSR range over several maturities, see Figures~\ref{fig:calib_2_factor_ssr} and \ref{fig:ssr_two_factor_no_calib}.

Finally, we assess the model’s consistency with the well-documented set of stylized facts observed in historical data. Empirical studies \cite*{chicheportiche2014fine, cont2001empirical, zumbach2009time, zumbach2010volatility} show that realistic volatility models should exhibit features such as volatility clustering, skewness, kurtosis, and leverage effect. We demonstrate that the two-factor Quintic model can reproduce these key statistical patterns, see  Figure~\ref{fig:stylized_facts_group}. In addition, it captures the elusive Zumbach effect that is often difficult for stochastic volatility models to capture, see Figure~\ref{fig:zumbach}. 


The numerical results presented in this paper demonstrate the potential of the two-factor Quintic OU model: not only does it achieve remarkable accuracy in capturing both the SPX and VIX volatility surfaces and the SSR, but it also aligns with the stylized facts of SPX historical time series. This makes the two-factor Quintic OU model a practical and effective solution to the long-standing calibration-dynamics trade-off in volatility modeling. {We provide a detailed implementation of the two-factor Quintic OU model:  \url{http://github.com/shaunlinz02/two_factor_quintic_ou}.} 

This paper is outlined as follows. Section \ref{2_fac_model} introduces the two-factor Quintic OU model, the pricing of SPX and VIX derivatives, and the computation of the model SSR. Section \ref{numerical_experiments} presents several calibration results of the two-factor Quintic OU model. Lastly, Section \ref{sf_all} compares the stylized facts observed in SPX market data with those generated by the two-factor Quintic OU model.

\section{The two-factor Quintic OU model}\label{2_fac_model}

The two-factor Quintic OU stochastic volatility model is defined by the following dynamics
\begin{equation}
  \begin{aligned}
    dS_t &= S_t\sigma_t \left(\rho dW_t + \sqrt{1-\rho^2} dW^{\perp}_t \right), \quad S_0>0,\\
    \sigma_t &= g_0(t)p(Z_t), \quad p(z) = \sum_{k=0}^5\alpha_k z^k,\\
    Z_t &=  \theta X_t+(1 - \theta)Y_t,\\
    X_t &= \int_0^t e^{-\lambda_x(t-s)}dW_s, \quad Y_t = \int_0^t e^{-\lambda_y(t-s)}dW_s,
  \end{aligned}\label{two_factor_Quintic}
\end{equation}
where $(W, W^{\perp})$ is a two-dimensional Brownian motion defined on a risk-neutral filtered probability space $(\Omega, \mathcal F,(\mathcal F_t)_{t\geq 0}, \mathbb Q )$ satisfying the usual conditions. The two-factor Quintic OU model has the following model parameters
$$
\Big( \lambda_x, \lambda_y, \theta, \rho, (\alpha_k)_{ 0 \leq k\leq 5} \Big),$$
with $\rho \in [-1,1]$, $\lambda_x, \lambda_y >0$ { such that $1/\lambda_x$ and $1/\lambda_y$ the characteristic time (in years)}, $\theta \geq 0$, $\alpha_k \in \R$. $g_0(\cdot)$ is a deterministic input curve that can be used to match the market term structure of volatility. For example, setting 
\begin{equation}
    g_0(t):=\sqrt{{\xi_0(t)}/{\E \left[p(Z_t)^2\right]}}, \quad t \geq 0, \label{g_o_expression}
\end{equation}
allows the model to match the term structure of the forward variance swap observable on the market, since
\[
    \mathbb E\left[ \int_0^T \sigma_t^2 dt \right] = \int_0^T \xi_0(t)dt, \quad T> 0.
\]
{\begin{remark}
Given the expression of $g_0(\cdot)$ in \eqref{g_o_expression}, one can normalize the polynomial $p(\cdot)$ in \eqref{two_factor_Quintic} by any non-zero constant without modifying the model dynamics. We can thus fix $\alpha_5 = 1$ to further reduce the number of model parameters to be calibrated. Furthermore, it can be shown that the underlying spot $S$ is a true martingale whenever $\alpha_5>0$ and $\rho \leq 0$, {as proved in the case of a one-factor Quintic model in the longer arXiv version of \cite* {jaber2022quintic}.}
\end{remark}
}

\begin{proposition} (Path-dependent property of $Z_t$)
    The process $Z_t$ admits the following dynamics
\begin{equation}\label{eq:Zdyn}
\begin{aligned}
    dZ_t &=\left(-\left( \lambda_x\theta + \lambda_y \left(1-\theta\right) \right)Z_t + \theta\left(1-\theta\right)\left(\lambda_y - \lambda_x\right)^2 \widehat Z_t\right)dt +  dW_t,
\end{aligned}
\end{equation}
with the path-dependent drift
\begin{equation}
\begin{aligned}\widehat Z_t = \int_0^te^{-\left(\lambda_x\left(1-\theta\right) + \lambda_y\theta\right)(t-s)}Z_sds.
\end{aligned}\label{z_hat}
\end{equation}
\end{proposition}

\begin{proof}
We first define the auxiliary process $\widetilde Z_t := X_t - Y_t,  
$
so that
\begin{equation}
\begin{pmatrix}
    X_t \\
    Y_t
\end{pmatrix}=\begin{pmatrix}
    1 & 1-\theta \\
    1 & -\theta
\end{pmatrix}
\begin{pmatrix}
    Z_t \\
    \widetilde Z_t
\end{pmatrix}.\label{rep_of_two_factors}
\end{equation}
Applying Itô's lemma on $Z_t$ and $\widetilde Z_t$ together with \eqref{rep_of_two_factors} yields
\begin{equation}
\begin{aligned}
    dZ_t =\left(-\left( \lambda_x\theta + \lambda_y (1-\theta) \right)Z_t +  \theta(1-\theta)(\lambda_y - \lambda_x) \widetilde Z_t\right)dt + dW_t,
\end{aligned}\label{y_t_representation}
\end{equation}
and
\begin{equation*}
\begin{aligned}
    d\widetilde Z_t &= \left(-(\lambda_x\left(1-\theta) + \lambda_y\theta\right) \widetilde Z_t+(\lambda_y - \lambda_x) Z_t\right)dt,
\end{aligned}
\end{equation*}
whose solution is
\begin{equation}
\widetilde Z_t =(\lambda_y - \lambda_x) \int_0^te^{-\left(\lambda_x\left(1-\theta\right) + \lambda_y\theta\right)(t-s)}Z_sds. \label{y_2_solution}
\end{equation}
Substituting \eqref{y_2_solution} into \eqref{y_t_representation} yields \eqref{eq:Zdyn}. 
\end{proof}

\begin{remark}
As it can be seen on \eqref{eq:Zdyn}, for small $t$, the process $Z$ behaves similarly to an OU process with mean reversion $ \lambda_x\theta + \lambda_y \left(1-\theta\right)$. For larger $t$, the contribution from the exponentially weighted sum of the path of $Z$ up to time $t$ becomes more significant, which helps in controlling the decay rate of term structures.
\end{remark}

\subsection{Explicit expression of the VIX}
In continuous time, $\mbox{VIX}_t^2$ is defined as
  \begin{equation}
    \mbox{VIX}_t^2 = \frac {100^2}{\Delta} \E \left[-2\log(S_{t+\Delta}/S_t)\mid \mathcal{F}_t \right] = \frac{100^2}{\Delta} \int_t^{t+\Delta}\xi_t(s) ds,\label{eq:vixdef}
  \end{equation}
with $\Delta= 30$ days, and $\xi_t(s):=\E \left[\sigma_s^2\mid \mathcal{F}_t\right]$
the forward variance process. We can now derive an explicit expression of $\mbox{VIX}_t^2$ for the two-factor Quintic OU model {thanks to the Markovianity and Gaussianity of the process  $(X,Y)$.} 

\begin{proposition}
Under the two-factor Quintic OU model, $ \mathrm{VIX}_t^2$ can be expressed as an explicit polynomial in $(X_t, Y_t)$:
\begin{align}\label{vix_expression}
        \mathrm{VIX}_t^2 &= \frac{100^2}{\Delta}\sum_{m=0}^{10} \sum_{\ell=m}^{10} \beta_{m,\ell}(t) X_t^{m}Y_t^{\ell-m}=:  h_t(X_t, Y_t),
\end{align}
where $\beta_{m,l}(t)$ is a deterministic function defined in \eqref{expression_beta}. 
\end{proposition}

\begin{proof}
By Markovianity, for $s\geq t$, we can decompose the process $Z_s$ as
\begin{equation*}
    \begin{aligned}
        Z_s &= \underbrace{\phantom{\int}\theta e^{-\lambda_x(s-t)} X_t + (1-\theta) e^{-\lambda_y(s-t)} Y_t}_{H_t^s} + \underbrace{{\int_t^s  \theta e^{-\lambda_x(s-u)} + (1-\theta) e^{-\lambda_y(s-u)}dW_u}}_{G_t^s},
    \end{aligned}\label{mark_exp_y}
\end{equation*}
where $H_t^s$ is $\mathcal{F}_t$ measurable, and $G_t^s$ is independent of $\mathcal{F}_t$. Applying the binomial theorem, we can write the forward variance process of the two-factor Quintic OU model $\xi_t(s)$ as
\begin{equation}
    \xi_t(s) = g_0^2(s) \sum_{k=0}^{10} (\alpha*\alpha)_k \E \left[ Z_s^k \mid \mathcal{F}_t\right]  = g_0^2(s)\sum_{k=0}^{10}\sum_{\ell=0}^k(\alpha*\alpha)_k  \binom{k}{\ell} (H_t^s)^{\ell} \E \left[ (G_t^s)^{k-\ell} \right],\label{fwd_var_process}
\end{equation}
with \begin{equation}
    \begin{aligned}
 (\alpha *\alpha)_k = & \sum_{j=0}^{k}\alpha_j\alpha_{k-j}, \quad  \binom{k}{i}= \frac{k!}{(k-i)!i!}.
    \end{aligned}
\end{equation}
The process $G_t^s$ is a centred Gaussian with variance
$$
\sigma^2_G(t, s) := \frac{\theta^2}{2\lambda_x}(1-e^{-2\lambda_x(s-t)})+\frac{(1-\theta)^2}{2\lambda_y}(1-e^{-2\lambda_y(s-t)})+\frac{2\theta(1-\theta)}{\lambda_x+\lambda_y} (1-e^{-(\lambda_x+\lambda_y) (s-t)}),
$$
so for all $p \in \mathbb{N}$, the quantity $\E \left[ (G_t^s)^{p} \right]$ is known explicitly. Substituting \eqref{fwd_var_process} into \eqref{eq:vixdef}, we have
\begin{equation}
\begin{aligned}
 \mbox{VIX}_t^2 &=  \frac{100^2}{\Delta} \int_t^{t+\Delta}\xi_t(s) ds = \frac{100^2}{\Delta}\sum_{k=0}^{10} \sum_{\ell=0}^{k}  (\alpha *\alpha)_k \binom{k}{\ell} \int_t^{t+\Delta} g_0^2(s) \mathbb{E}\left[(G_t^s)^{k-\ell} \right] (H_t^s)^{\ell}ds.
\end{aligned}\label{vix_sqaured_full}
\end{equation}
Re-applying the binomial theorem on the term $(H_t^s)^{\ell}$ and re-arranging the sums, we obtain the representation \eqref{vix_expression}, with 
\begin{equation}
    \begin{aligned} \beta_{m,\ell}(t) := &\binom{\ell}{m} \sum_{k=\ell}^{10} (\alpha *\alpha)_k \binom{k}{\ell} \theta^{m}(1-\theta)^{l-m}\int_t^{t+\Delta} g_0^2(s) \mathbb{E}\left[(G_t^s)^{k-l} \right] e^{-(m\lambda_x + (l-m)\lambda_y)(s-t)}ds
\end{aligned}\label{expression_beta}
\end{equation}
\end{proof}

\subsection{Pricing of VIX derivatives}
From \eqref{vix_expression}, we know that  $\mbox{VIX}_t^2$ is polynomial in $(X_t,Y_t)$ which is Gaussian with the law
\[
\begin{pmatrix}
    X_t \\
    Y_t
\end{pmatrix} \sim \mathcal{N} (\mu, \Sigma), \quad 
\mu = \begin{pmatrix}
    0 \\
    0
\end{pmatrix}, \quad \Sigma = \begin{pmatrix}
    \sigma_{xx} & \sigma_{xy} \\
    \sigma_{yx} & \sigma_{yy} 
\end{pmatrix},
\]
where
\[
\sigma_{ij} = \frac{1-e^{-(\lambda_i+\lambda_j)t}}{(\lambda_i+\lambda_j)}, \quad (i,j) \in \{x,y\}^2.
\]
Pricing of VIX derivatives with payoff function $\Phi$ is thus immediate by integrating the payoff with respect to the Gaussian density, that is
\begin{equation}\label{vix_derivative_pricing}
\E \left[\Phi(\mbox{VIX}_t) \right] = \E \left[\Phi\left(\sqrt{h_t(X_t,Y_t)}\right) \right] = \int_{\R^2} \Phi\left(\sqrt{h_t\left( x, y\right)}\right)f(x, y) {dxdy},
\end{equation}
where
\[
f(x, y) = \exp\left(-\frac{1}{2}(x,y) \Sigma^{-1} (x,y)^{\top}\right)/\left(2\pi \sqrt{\mbox{det}(\Sigma)}\right)
\]
is the two-dimensional Gaussian density. 

\begin{example}
 To price VIX futures prices, we set $\Phi(v) = v$; to price VIX vanilla call price, we set $\Phi(v) = (v-K)^{+}$.
\end{example}

The integral \eqref{vix_derivative_pricing} can be computed using various numerical cubature techniques. In particular, one can take advantage of the underlying Gaussianity of $\mbox{VIX}$ and make use of the quantization cubature \cite*{abi2022joint,pages2003optimal} for efficient computation of (conditional) expectations of functions of Gaussian variables. Compared to most cubature techniques that require double sums to approximate the double integral in \eqref{vix_derivative_pricing}, the quantization cubature involves only a single sum. Figure \ref{fig:quant_cub} shows numerical convergence of the VIX ATM implied volatility across different maturities using quantization cubature. We can see that the quantization cubature is highly accurate, and converges quickly to the Monte Carlo confidence interval with as few as 500 points to evaluate.
\vspace{-0.3cm}

  \begin{figure}[H]
    \centering    \includegraphics[width=0.7\textwidth]{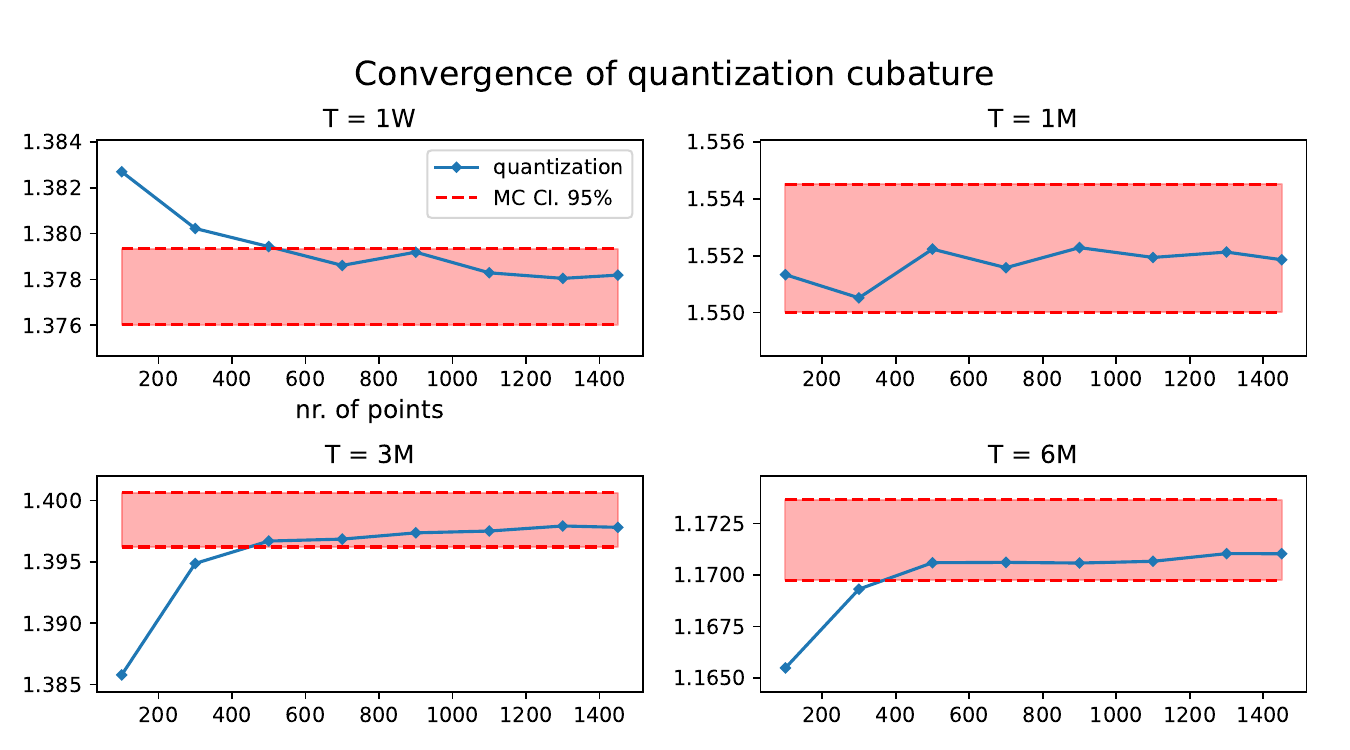}%
    \vspace{-0.5cm}
    \caption{Numerical convergence of the quantization cubature (in blue), vs.~the confidence interval generated by Monte Carlo estimator with 20,000,000 simulations. The parameters used here are the same as the calibrated parameters in Section \ref{sec:jc_spx_ssr}, with flat forward variance $\xi_0(t) = 0.03$.} 
    \label{fig:quant_cub} 
  \end{figure}

\subsection{Pricing of SPX derivatives}
The pricing of SPX derivatives is straightforward. The instantaneous volatility $\sigma$ can be simulated exactly through $(X, Y)$ due to their Gaussianity. Once the volatility process is obtained, the log-price of the underlying can be simulated using a standard Euler scheme, combined with variance reduction techniques to price SPX derivatives. 
We refer the reader to the detailed Python code implementation here: {\url{http://github.com/shaunlinz02/two_factor_quintic_ou}}.

\subsection{Model Skew-stickiness ratio}\label{ssr_mc}
The model SSR, $\mathcal{R}_{t,T}$ is defined in \cite*{bergomi2015stochastic} as
\begin{equation}
\mathcal{R}_{t, T} := \frac{1}{\mathcal{S}_{t,T}}\frac{\partial_t \langle \widehat \sigma^T, \log S \rangle_t}{\partial_t \langle \log S, \log S \rangle_t},  \label{def_ssr}
\end{equation}
where $\widehat \sigma^T_t$ denotes the ATM (forward) Black-Scholes implied volatility for a European option written on the underlying $S_t$ with maturity $T-t$, and $\mathcal{S}_{t,T}$ is the ATM skew. { Here, $\langle \cdot, \cdot \rangle$ denotes the quadratic covariation.} The ratio in \eqref{def_ssr} can be interpreted as the instantaneous change of the ATM implied volatility over the instantaneous change of the log-price, normalized by the ATM skew.

The model SSR of the two-factor Quintic OU model can be computed using a combination of Monte Carlo and finite-difference techniques as prescribed in \cite*{bergomi2015stochastic}. Given the Markovianity of the two-factor Quintic OU model, the ATM implied volatility process $\widehat \sigma_t^T$ is a function in $t$ and $(X_t, Y_t)$:
\[
\widehat \sigma_t^T := \widehat \sigma^T(t, X_t, Y_t).
\]
Assume that $\widehat \sigma^T$ is $\mathcal{C}^{1,2,2}$, we can apply Itô:
\[
d\widehat \sigma_t^T = \partial_{x}\widehat \sigma^T(t, X_t, Y_t)dX_t + \partial_{y}\widehat \sigma^T(t, X_t, Y_t)dY_t + \ldots dt
\]
so that
\[
\partial_t \langle \widehat \sigma^T, \log S \rangle_t = \rho \sigma_t  \left( \partial_{x}\widehat \sigma^T(t, X_t, Y_t) + \partial_{y}\widehat \sigma^T(t, X_t, Y_t) \right), 
\]
and
\[
\partial_t \langle \log S, \log S \rangle_t = \sigma_t^2.
\]
Substituting the above into \eqref{def_ssr}, we get
\begin{equation*}
    \begin{aligned}
    \mathcal{R}_{t,T} &= \frac{1}{\mathcal{S}_{t, T}}\frac{\partial_t \langle \widehat \sigma^T, \log S \rangle_t}{\partial_t \langle \log S \rangle_t}\\
    &=  \frac{\rho  \left( \partial_{x}\widehat \sigma^T(t, X_t, Y_t) + \partial_{y}\widehat \sigma^T(t, X_t, Y_t) \right)}{\sigma_t\mathcal{S}_{t,T}}\\
    &= \lim_{\epsilon \xrightarrow[]{}0}  \frac{\widehat \sigma^T(t, X_t + \epsilon\frac{\rho}{\sigma_t}, Y_t +\epsilon\frac{\rho}{\sigma_t}) - \widehat \sigma^T(t, X_t,Y_t)}{\epsilon \mathcal{S}_{t,T}}.
    \end{aligned}
\end{equation*}

\begin{remark}
    In the case of the one-factor Quintic OU model, i.e.~$\theta \in \{0,1\}$ and the one-factor Bergomi model, the model SSR can be computed efficiently via Fourier techniques using the semi-explicit expression of the joint Fourier-Laplace transform of the log price and integrated variance, see \cite*{jaber2024fourier}.
\end{remark}

\section{Calibration results}\label{numerical_experiments}

\subsection{Joint calibration of SPX and SSR term structure}\label{sec:jc_spx_ssr}

We calibrate the two-factor Quintic OU model to the term structure of SPX ATM implied volatility, skew, and SSR as of 6 May 2024 using data points reconstructed from various figures in \cite*{bourgey2024smile}, {where the empirical SSR is estimated using a 60-day rolling window before the calibration date as appeared in the top left sub-figure of Figure 14}. Figure \ref{fig:spx_term_struc} demonstrates the fit of the two-factor Quintic OU model, with the following calibrated parameters $\lambda_x = 33.754,\: \lambda_y = 2.027,\: \theta = 0.678,\: \rho = -0.588,\: (\alpha_0,\:\alpha_1,\: \alpha_2,\: \alpha_3,\: \alpha_4) = (0.0025,\:  0.009,\: -0.0594,\: -0.0328,\  0.3239)$.

  \begin{figure}[H]
    \centering    \includegraphics[width=1.02\textwidth]{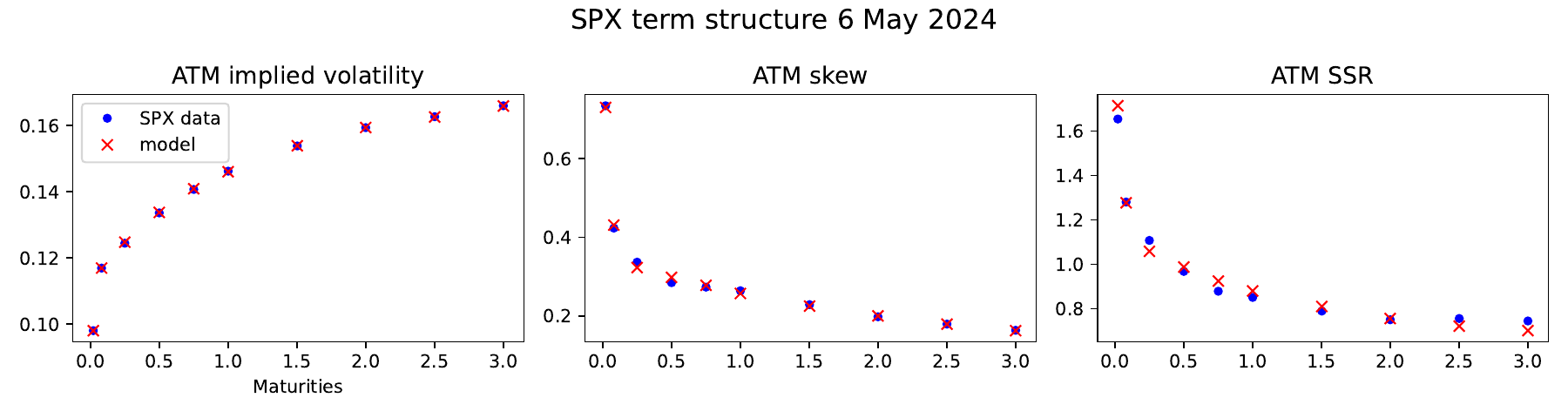}%
    \vspace{-0.5cm}
    \caption{The term structure of the SPX ATM implied volatility, skew and SSR as of 6 May 2024 (in blue), calibrated by the two-factor Quintic OU model (in red).} 
    \label{fig:spx_term_struc} 
  \end{figure}

In contrast with the initial one-factor Quintic model \cite*{jaber2022quintic}, where the parameters $\alpha_2$ and $\alpha_4$ were set to zero, we allow these parameters to be free. This allows the model to capture the term structure of the SSR to produce richer dynamics and fits better than classical (multifactor) stochastic volatility models such as Heston and Bergomi. The remarkable fit shown in Figure \ref{fig:spx_term_struc} suggests that stochastic volatility models are capable of capturing the term structure of SPX volatility surface and its dynamics, contrary to popular beliefs in the quantitative finance community.

{We can compare the fit in Figure \ref{fig:spx_term_struc} to the calibrated two-factor Bergomi model, a popular model among practitioners. Figure \ref{fig:two_factor_bergomi_ssr} shows that the two-factor Bergomi model is unable to produce the desired SSR term structure, even though it is capable of matching the ATM implied volatility and ATM skew term structure. Similar results were obtained in [Figure 14, \cite*{bourgey2024smile}]. 

  \begin{figure}[H]
    \centering    \includegraphics[width=1.02\textwidth]{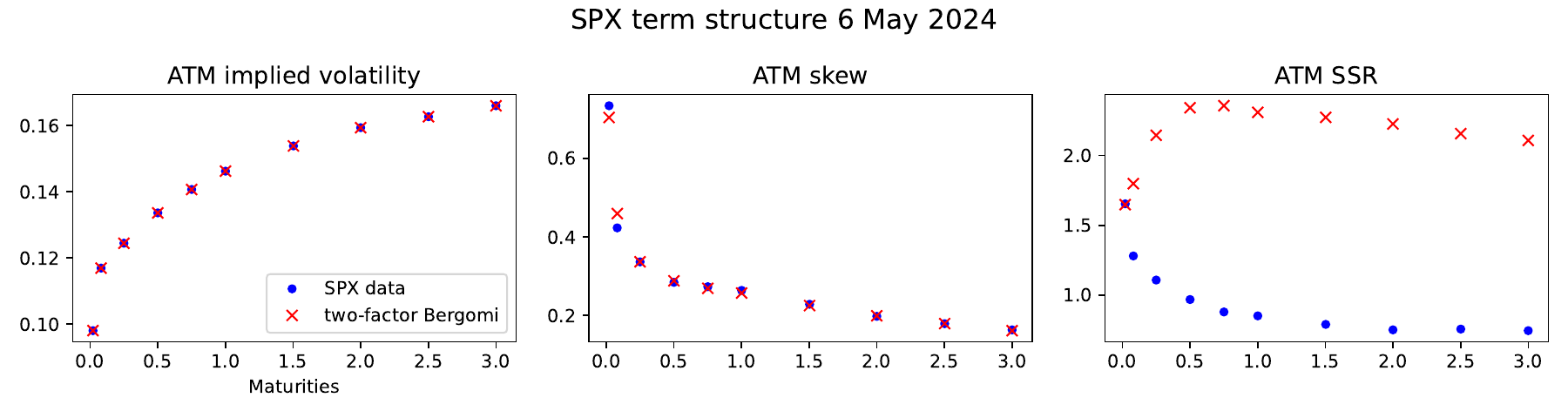}%
    \caption{The term structure of the SPX ATM implied volatility, skew, and SSR as of 6 May 2024 (in blue), calibrated by the two-factor Bergomi model (in red). 
    } 
    \label{fig:two_factor_bergomi_ssr} 
  \end{figure}

In addition, the two-factor Quintic OU model can generate nearly identical SPX smiles with very different SSR profiles. In Appendix \ref{diff_ssr_profile}, we plot the calibrated SPX smiles for different parameter sets, along with their respective SSR term structures and the associated polynomial function $p$ defined in \eqref{two_factor_Quintic} to highlight the flexibility of the two-factor Quintic OU model.
  
}

\subsection{Joint calibration of SPX and VIX}\label{sec:jc_no_ssr}
 \cite*{jaber2022quintic}  showed that the one-factor Quintic OU model is capable of jointly calibrating the SPX and VIX volatility smiles for short maturities between one week and three months. In this section, we show how adding a second factor allows the quintic OU model to jointly calibrate a larger range of maturities. At first, we set $\alpha_2 = \alpha_4 = 0$, similar to the one-factor Quintic OU model to keep the model parsimonious. With $\alpha_5 = 1$ fixed as per Section \ref{2_fac_model}, there are seven model parameters 
\[
(\lambda_x, \lambda_y, \theta, \rho, \alpha_0, \alpha_1, \alpha_3)
\]
plus the forward variance curve $\xi_0(t)$ to calibrate. Figure \ref{fig:calib_2_factor} shows that the two-factor Quintic OU model is capable of fitting the full SPX \& VIX smiles for maturities between one week and two years on 23 October 2017, with the following calibrated parameters $\lambda_x = 31.8,\: \lambda_y = 0.659,\:  \theta = 0.964,\: \rho = -0.765,\: (\alpha_0,\:\alpha_1,\:\alpha_3) = 0.0004,\: 0.0046,\: 
 0.0096)$.
  \begin{figure}[H]
    \centering
    \includegraphics[scale=0.45]{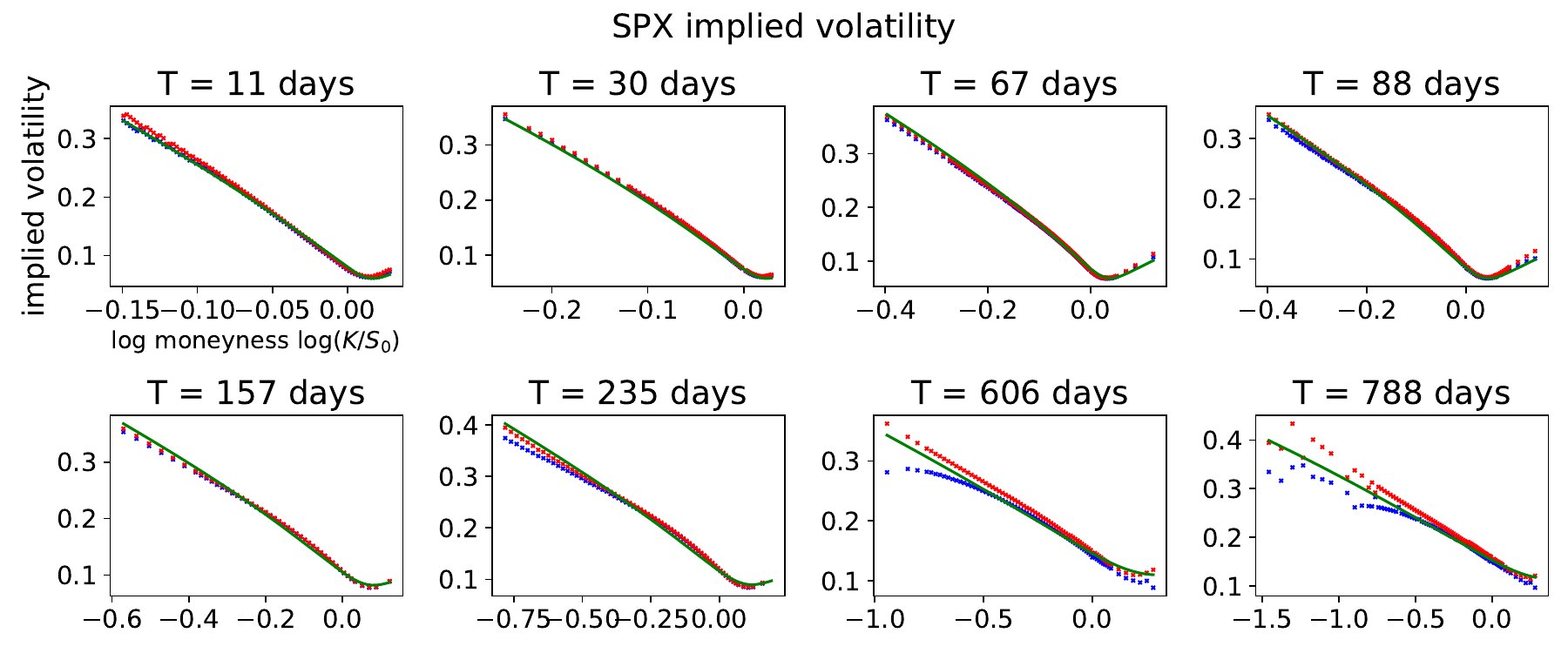}
    \includegraphics[scale=0.45]{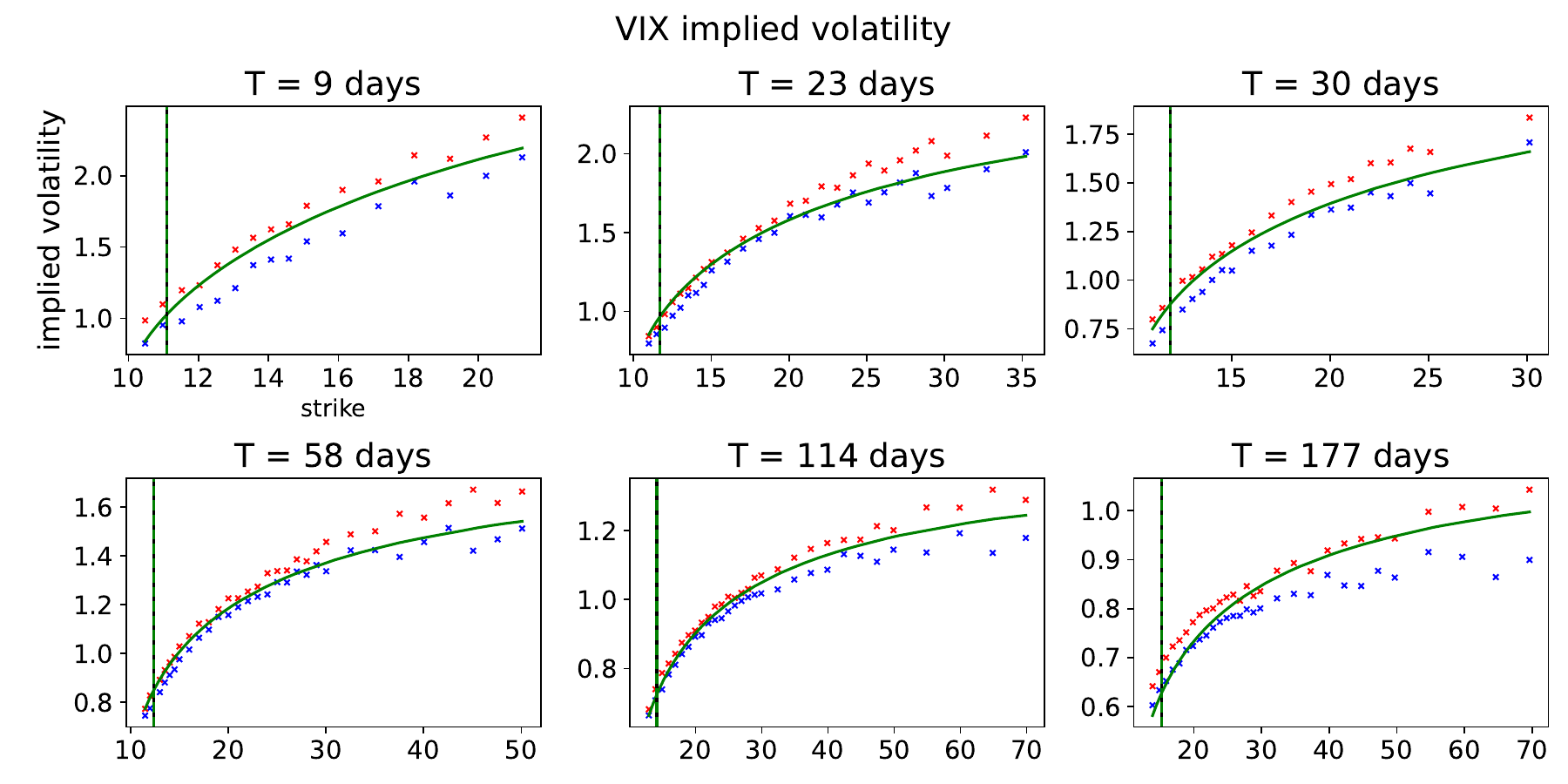}
    \caption{SPX \& VIX smiles (bid/ask in blue/red dots) and VIX futures (vertical black lines) on 23 October 2017, jointly calibrated by the two-factor Quintic OU model (green lines) {without SSR penalisation}.}
    \label{fig:calib_2_factor}
  \end{figure}

We can now compute the model SSR using the calibrated parameters as described in Section \ref{ssr_mc}. In order to compare model SSR with empirical observation, we computed the time series of SPX SSR over the period 2012 to 2022, which mostly fluctuates between 0.9 and 2.0 as shown in Figure \ref{fig:ssr_time_series}. This is in line with other independent computations in the literature, see [Figure 6, \cite*{gatheral2023computing}], and  [Figure 18, \cite*{bourgey2024smile}]. The left-hand side graph of Figure \ref{fig:ssr_two_factor_no_calib} shows the level of the model SSR using the above calibrated parameters, and one immediately sees that the SSR generated by this set of calibrated parameters is significantly below the market level shown in Figure \ref{fig:ssr_time_series}.

  \begin{figure}[H]
    \centering    \includegraphics[width=0.9\textwidth]{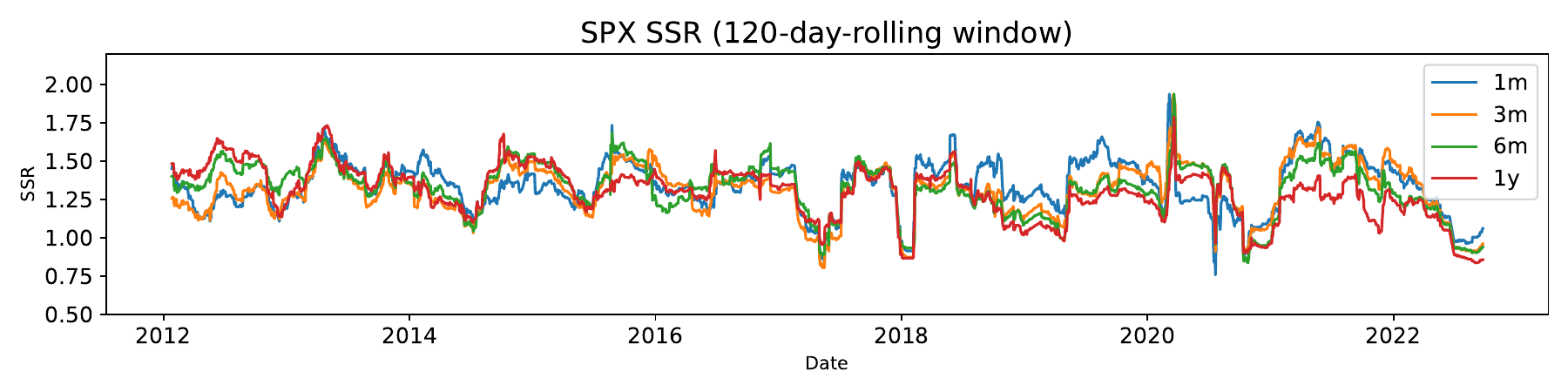}%
    \vspace{-0.5cm}
    \caption{Empirical time series of the SSR from 2012 to 2022, computed using a 120-day rolling window with data points outside two standard deviations filtered out.} 
    \label{fig:ssr_time_series} 
  \end{figure}
\vspace{-0.5cm}

\subsection{Joint calibration of SPX, VIX and SSR}\label{jc_3}

To ensure that the two-factor Quintic OU model respects the range of empirical SPX SSR while remaining consistent with the SPX and VIX smiles, {we bring back $\alpha_2$ and $\alpha_2$ as model parameters to be calibrated, and include an additional loss function during the joint calibration step, penalizing any model parameters} that would result in a model SSR outside the interval $[0.9,  2.0]$.

After taking into account the model SSR in the calibration procedure, Figure \ref{fig:calib_2_factor_ssr} shows that the two-factor Quintic OU model is still able to achieve a decent fit to the joint SPX and VIX smiles. In addition, the right-hand side graph of Figure \ref{fig:ssr_two_factor_no_calib} shows that the SSR of the two-factor Quintic OU model is now also consistent with the usual range observed on the market. The calibrated parameters are $\lambda_x = 35.2,\: \lambda_y = 0.623,\: \theta = 0.94,\: \rho = -0.769$,\:  $(\alpha_0,\:\alpha_1,\:\alpha_2,\:\alpha_3,\:\alpha_4) = (0.0004,\: 0.0038,\: 0.0004,\: 0.0085,\: 0.0005)$.

  \begin{figure}[H]
    \centering
    \includegraphics[scale=0.45]{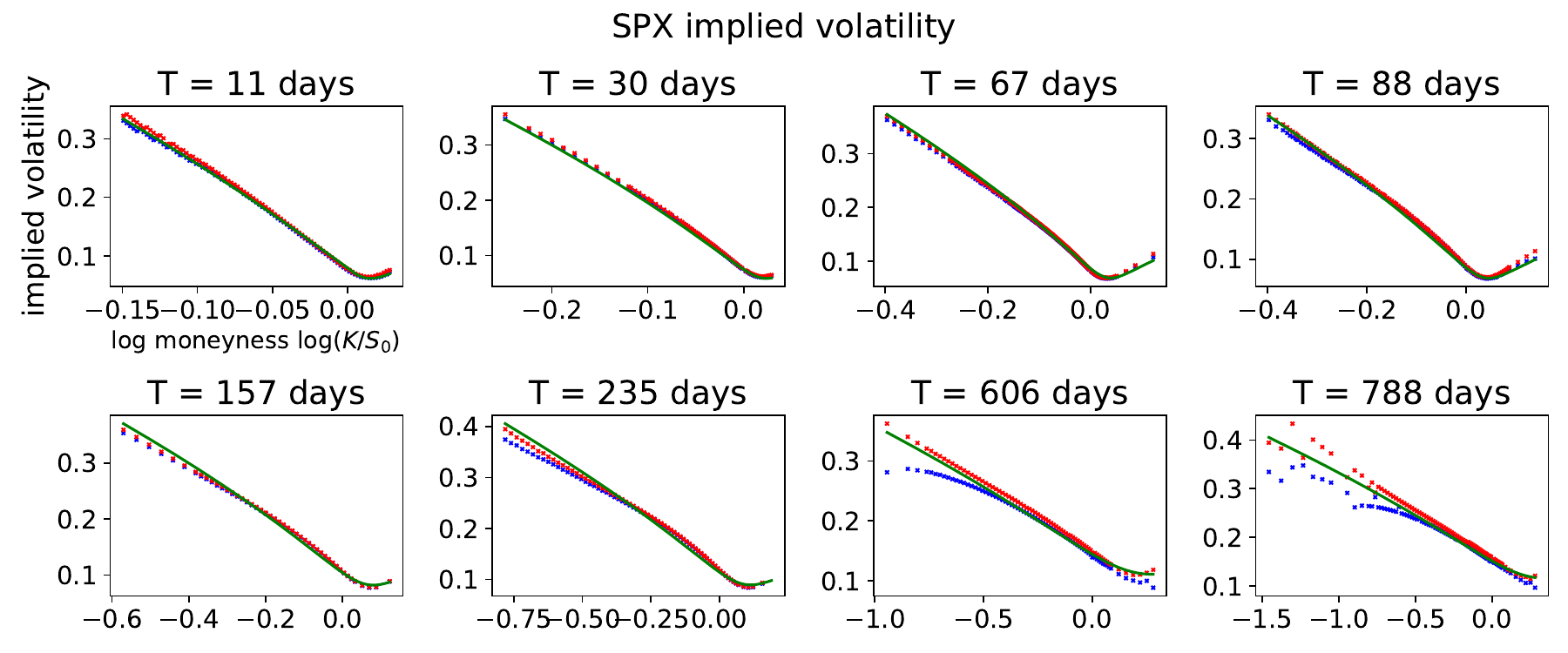}
    \includegraphics[scale=0.45]{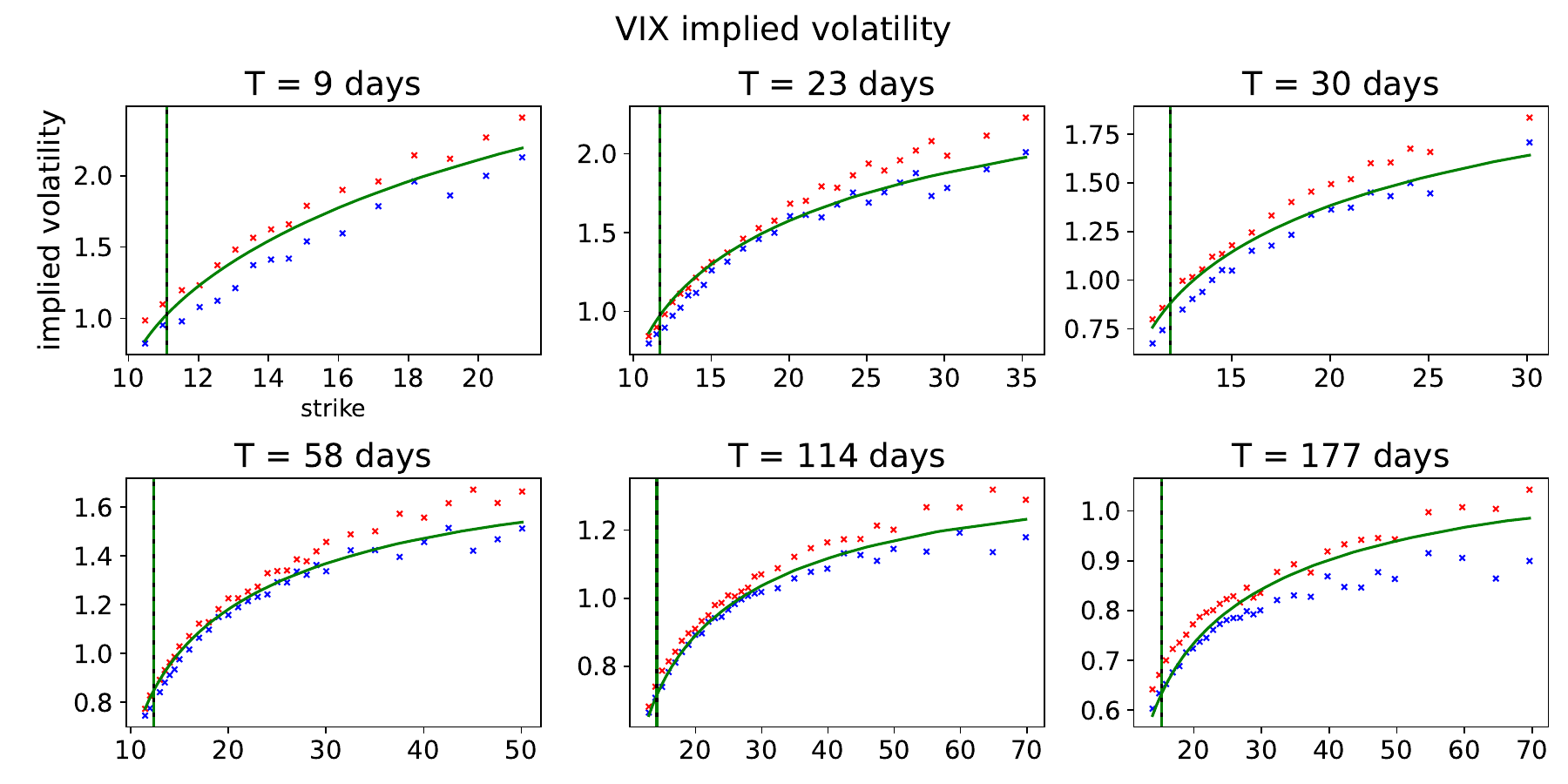}
    \caption{SPX \& VIX smiles (bid/ask in blue/red dots) and VIX futures (vertical black lines) on 23 October 2017, jointly calibrated by the two-factor Quintic OU model (in green) {with SSR penalisation}.}
    \label{fig:calib_2_factor_ssr}
  \end{figure}

\begin{figure}[H]
\centering
\includegraphics[width=0.45\textwidth]{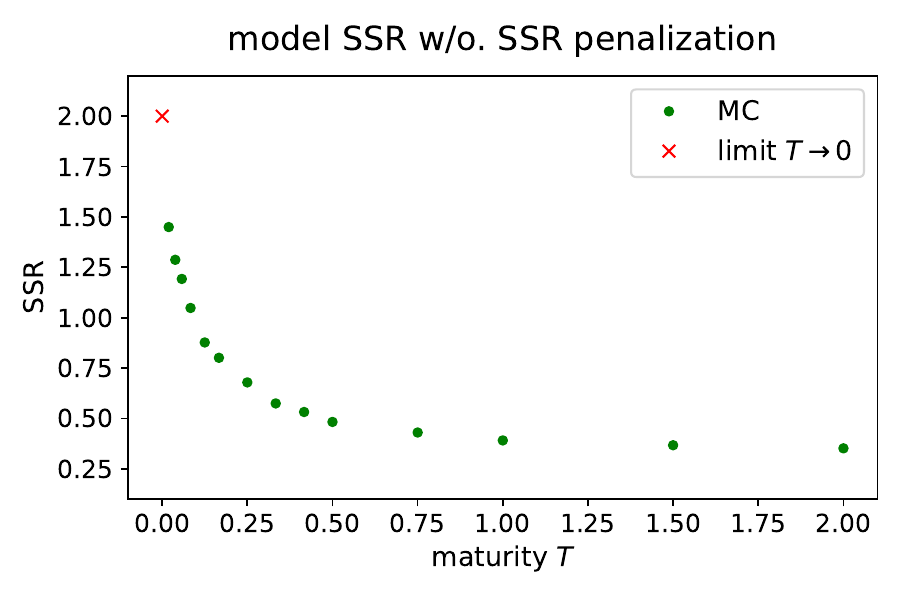}%
\includegraphics[width=0.45\textwidth]{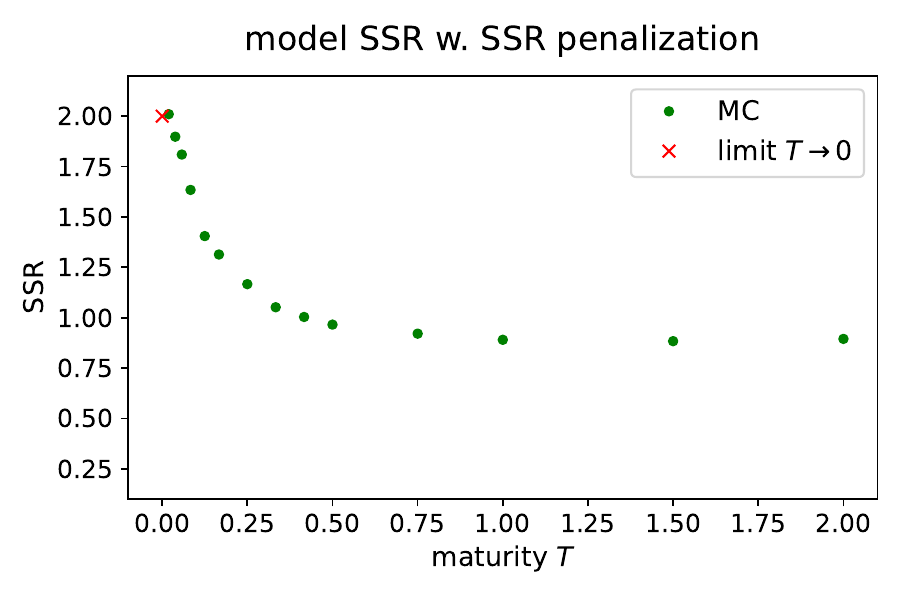}
\vspace{-0.4cm}

\caption{SSR of the two-factor Quintic OU model computed by finite difference and Monte Carlo. The left-hand side graph is the SSR of the two-factor Quintic OU model jointly calibrated to SPX and VIX smiles ({as per Section \ref{sec:jc_no_ssr}  with $\alpha_2$ and $\alpha_4$ kept at zero}). The right-hand side graph is the SSR of the two-factor Quintic OU model jointly calibrated to SPX and VIX smiles, as well as the SSR (({\ref{jc_3} with $\alpha_2$ and $\alpha_4$ kept at zero})).} \label{fig:ssr_two_factor_no_calib}
\end{figure}

\section{Stylized facts of the two-factor quintic OU model}\label{sf_all}

Besides fitting market smiles, it is also desirable for a model to reproduce market stylized facts. These are statistical properties of financial time series that reflect a broad consensus among researchers and practitioners regarding the behavior of financial markets, established through extensive empirical studies over time, such as
\cite*{chicheportiche2014fine, cont2001empirical,  zumbach2010volatility}. In this section, we show how the two-factor Quintic OU model can produce several key stylized facts that are consistent with market observations, including the elusive Zumbach effect that is often difficult for traditional Markovian stochastic volatility models to capture.


To compute the stylized facts of the SPX, we used the actual S\&P500 realized volatility and closing prices between 2000 and 2017. For the two-factor Quintic OU model, we first simulate a trajectory of the spot volatility process $\sigma_t$, with 5-minute time steps to allow for the construction of daily realized volatility over a time horizon of 20 years. We then simulate the trajectory of the log price to compute the model stylized facts. The parameters used to generate the spot volatility and log price trajectories of the two-factor Quintic OU model are: $\lambda_x = 35.2,\: \lambda_y = 0.623,\: \theta = 0.86,\: \rho = -0.7,\: (\alpha_0,\:\alpha_1,\:\alpha_2,\:\alpha_3,\:\alpha_4) = (0.0509,\: 0.4475,\: 0.0504,\: 0.9981,\: 0.059)$. {This parameter set is similar to that under Section \ref{jc_3}, with slight tweaks by hand to fit better the {magnitude of the} observed stylized facts}.

Figure \ref{fig:stylized_facts_group} shows several key stylized facts of SPX vs.~the same stylized facts produced by the two-factor Quintic OU model. We refer the reader to \cite*{cont2001empirical} for the precise definition of volatility clustering, skewness, kurtosis, and leverage effect, which we have followed in this Section. Based on Figure \ref{fig:stylized_facts_group}, it appears that the stylized facts produced by the two-factor quintic OU model are largely consistent, at least qualitatively, with that of the SPX. This is a promising outcome for a standard stochastic volatility model with only two Markovian factors.

  \begin{figure}[H]
    \centering
    \subfigure{
    \includegraphics[width=0.22\textwidth]{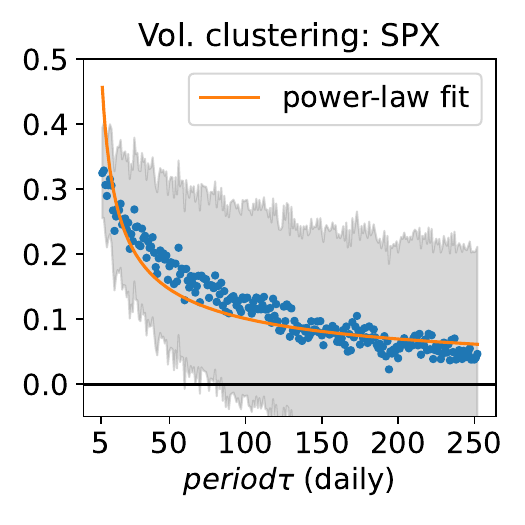}
    }
    \subfigure{
        \includegraphics[width=0.22\textwidth]{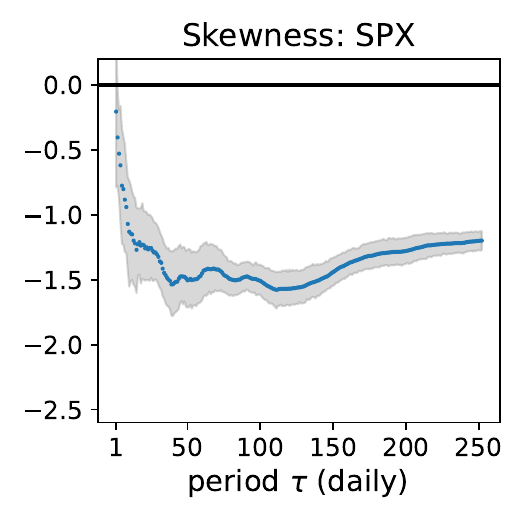}
    }
    \subfigure{
        \includegraphics[width=0.22\textwidth]{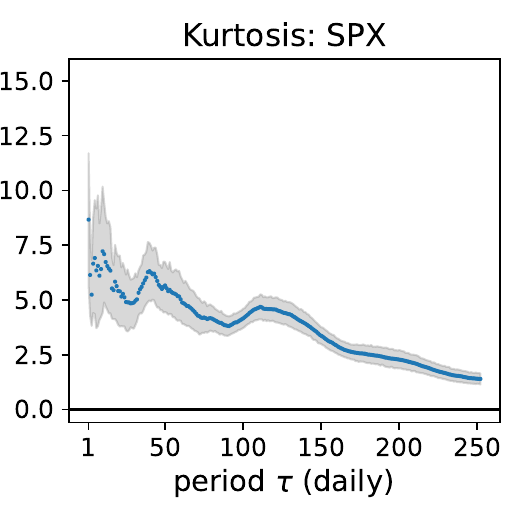}
    }
    \subfigure{
        \includegraphics[width=0.22\textwidth]{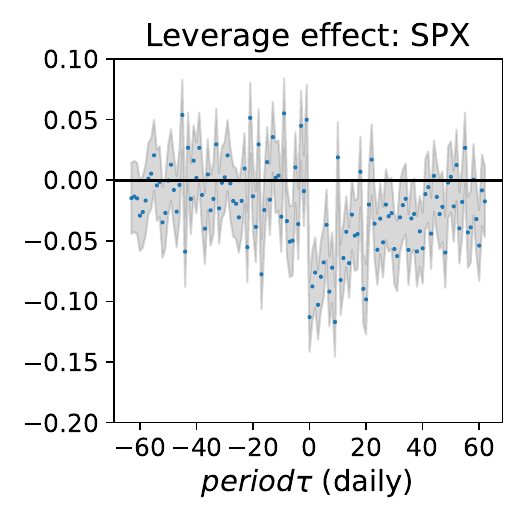}
    }
    \vspace{-0.5cm}

    \subfigure{
    \includegraphics[width=0.22\textwidth]{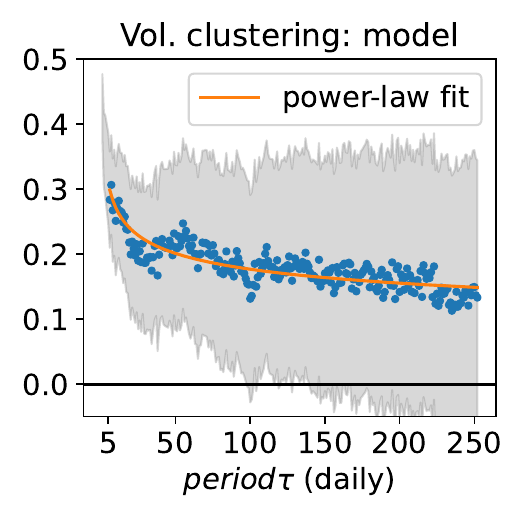}
    }
    \subfigure{
        \includegraphics[width=0.22\textwidth]{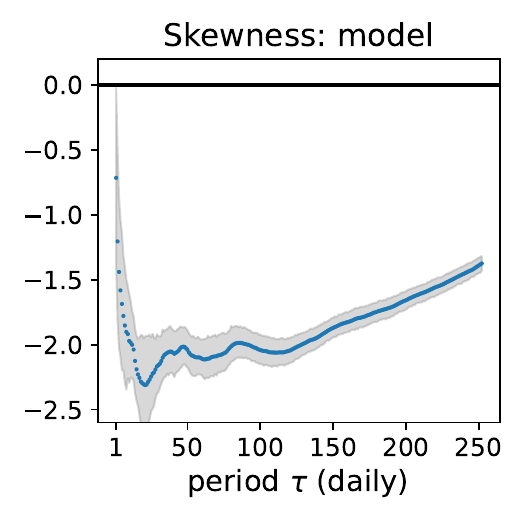}
    }
    \subfigure{
        \includegraphics[width=0.22\textwidth]{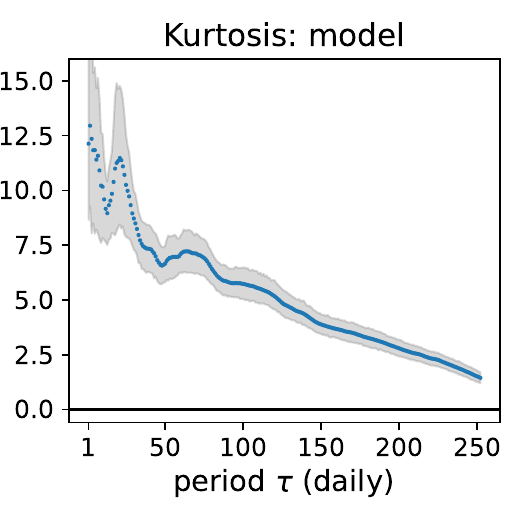}
    }
    \subfigure{
        \includegraphics[width=0.22\textwidth]{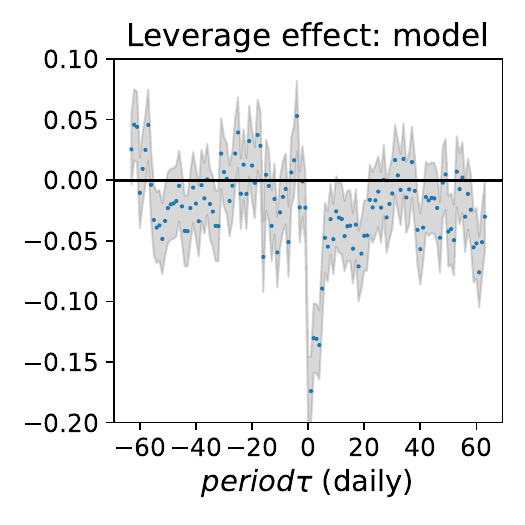}
    }

    \caption{Several stylized facts of SPX for the period between 2000 and 2017 vs.~the two-factor Quintic OU model. The shaded gray area represents 95 percent confidence interval. {The confidence interval for the Skewness and Kurtosis estimator is computed via bootstrapping with 1,000 iterations.}}
    \label{fig:stylized_facts_group}
\end{figure}


We now look at the  Zumbach effect, which measures the asymmetry of the predictive power between squared returns and (realized) volatility. Empirical studies show that past squared returns are better at predicting future volatility than past volatility at predicting future squared returns, see \cite*{zumbach2009time,zumbach2010volatility}. To measure this asymmetry, we use the estimator from [Equation 25, \cite*{chicheportiche2014fine}]. First, we define
\begin{align*}
\mathcal{C}(\tau) &:= \langle (\Tilde \sigma_t^2 - \langle \Tilde \sigma_t^2\rangle) r_{t-\tau}^2\rangle, \quad \Tilde \rho(\tau) := \frac{\mathcal{C}(\tau)}{\sqrt{\langle (\Tilde \sigma_t^2 - \langle \Tilde \sigma_t^2\rangle)^2 ( r_{t-\tau}^2 - \langle r_{t-\tau}^2\rangle)^2 \rangle}},
\end{align*}
with $\Tilde \sigma_t$ the (annualized) daily realized volatility based on the 5-minute intra-day returns, $r_{t}$ the daily log-return,  and $\langle \cdot \rangle$ a notation commonly used by physicists to denote the empirical average of a finite sample. The Zumbach effect is defined as
\begin{equation}
\mathcal{Z}(\tau) := \sum_{i=1}^{\tau}\Big(\Tilde \rho \left(\tau_i\right) - \Tilde \rho \left(-\tau_i\right) \Big).
\end{equation}
Figure \ref{fig:zumbach} shows the asymmetry of the Zumbach effect between squared returns and realized volatility indeed exists in the two-factor Quintic OU model, {and is consistent with the Zumbach effect of SPX in terms of cumulated asymmetry $\mathcal{Z}$} and the speed of decay of $\Tilde \rho$. One can expect even better results when model parameters are specifically fitted to the stylized facts of SPX.

  \begin{figure}[H]
    \centering
    \subfigure{
    \includegraphics[width=0.35\textwidth]{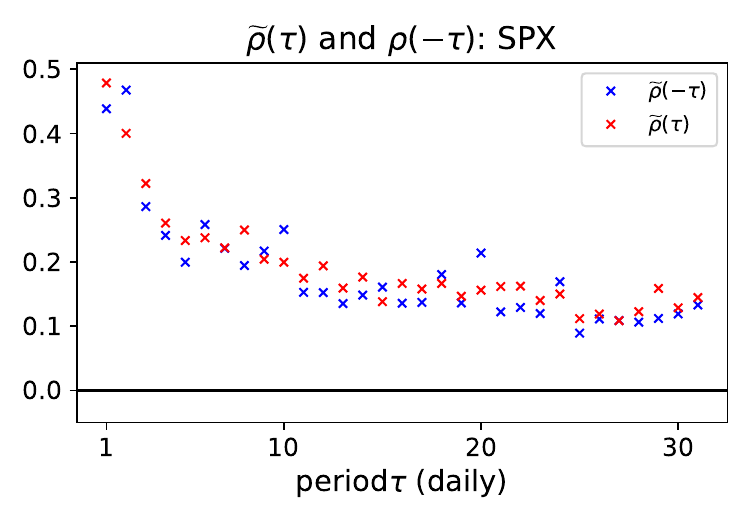}
    }
    \subfigure{
        \includegraphics[width=0.35\textwidth]{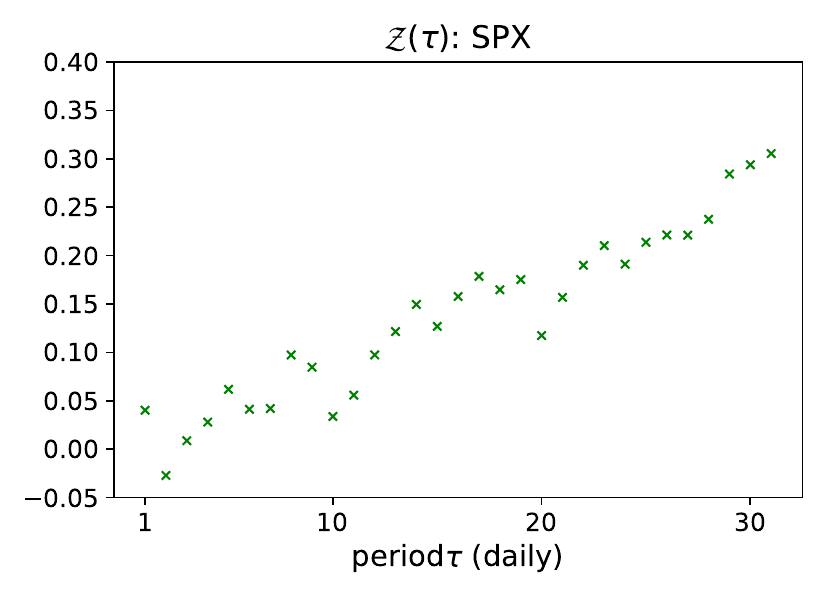}
    }
    \subfigure{
        \includegraphics[width=0.35\textwidth]{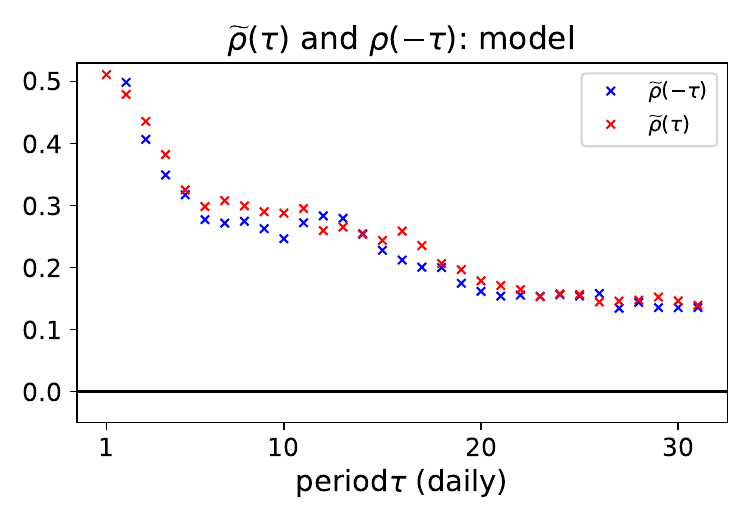}
    }
    \subfigure{
        \includegraphics[width=0.35\textwidth]{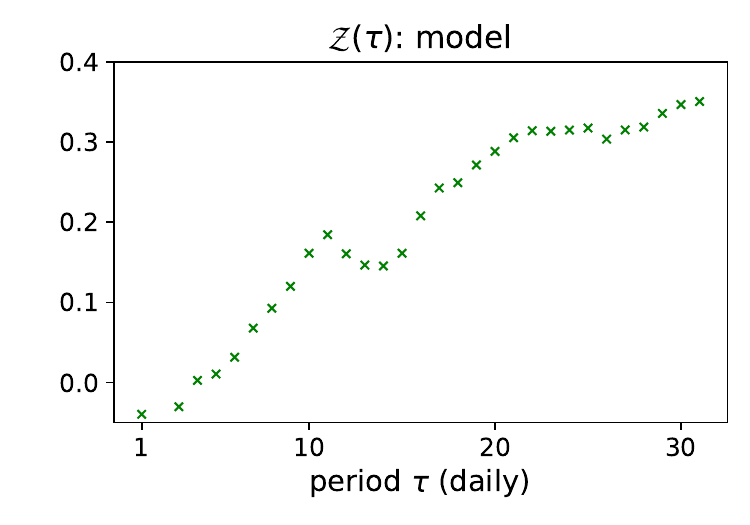}
    }
    \caption{Zumbach effect of the SPX for the period between 2010 and 2017 vs.~the two-factor Quintic OU model.}
    \label{fig:zumbach}
\end{figure}

{
\appendix



\section{Same smile, different SSR term structures}\label{diff_ssr_profile}
  \begin{figure}[H]
    \centering    \includegraphics[width=0.6\textwidth]{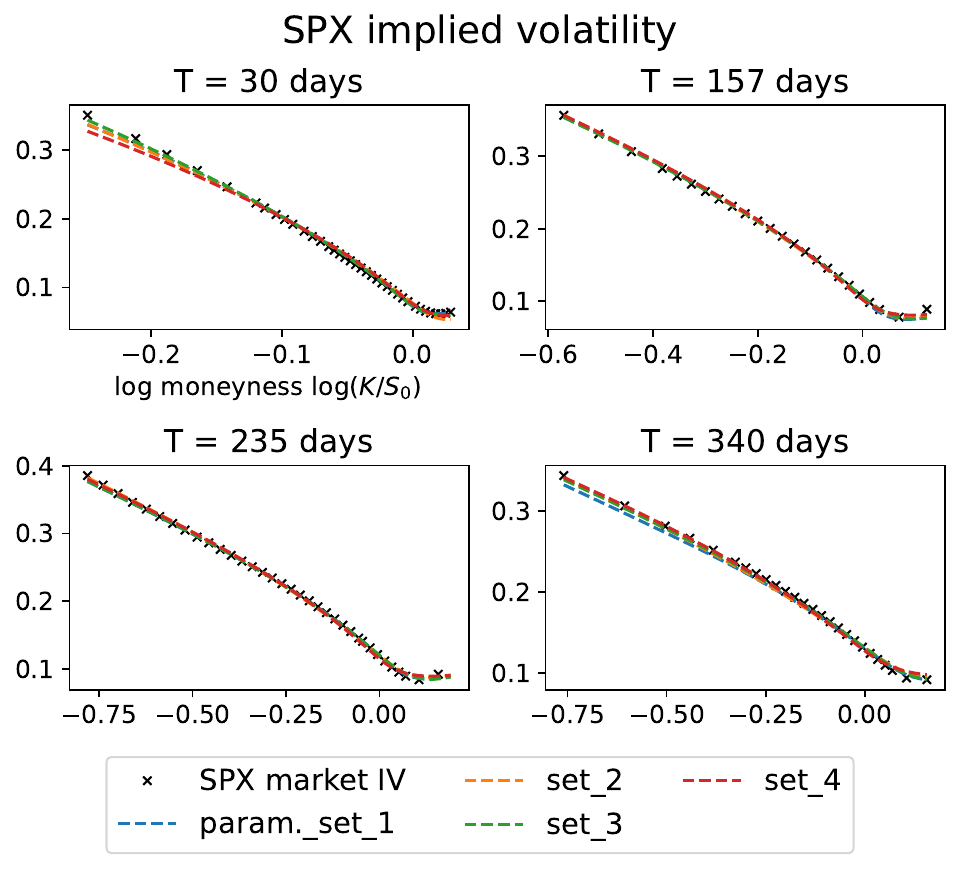}%
    \caption{SPX smiles generated by the two-factor Quintic OU model with different parameter sets, calibrated to the SPX smiles as of 23 October 2017, with maturities ranging from 1 month to 1 year.}
    \label{fig:two_factor_quintic_smile} 
  \end{figure}
\vspace{-0.5cm}

  \begin{figure}[H]
    \centering    \includegraphics[width=0.8\textwidth]{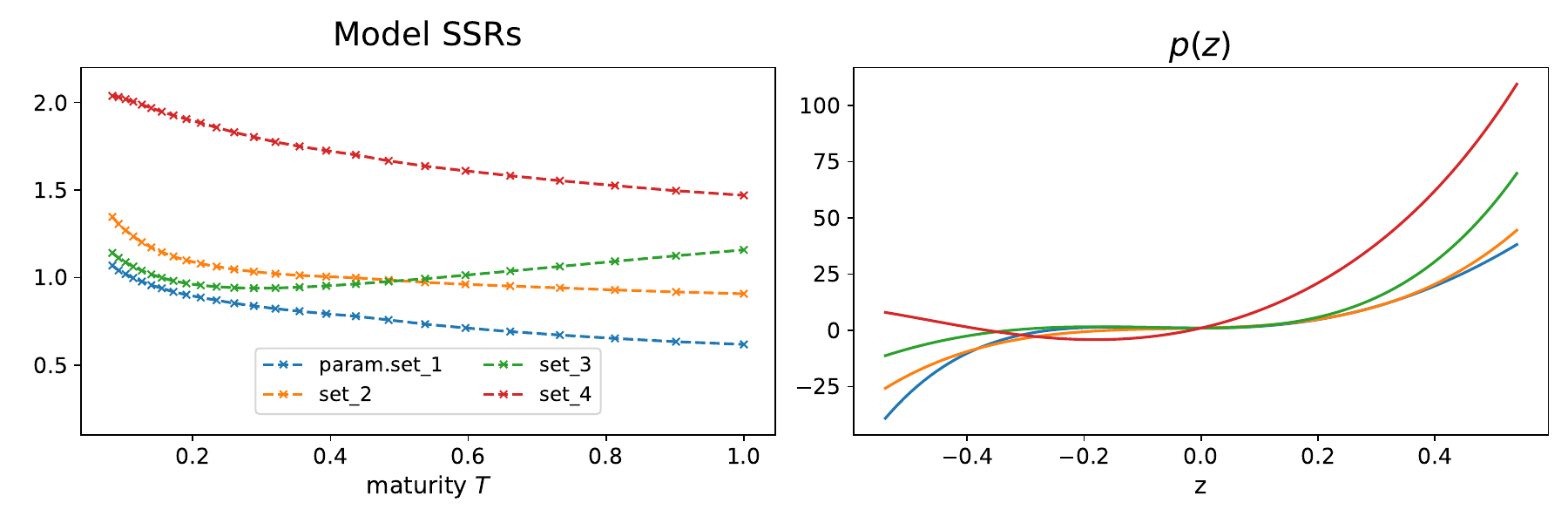}%
    \caption{The term structure of model SSRs (left) and the polynomial function $p(z)$ (right), generated by the two-factor Quintic OU model using the parameter sets as per in Figure \ref{fig:two_factor_quintic_smile}.} 
    \label{fig:ssr_and_p_curve_multi} 
  \end{figure}

}




\bibliographystyle{plainnat}
\bibliography{bibl.bib}

\end{document}